\documentclass[11pt,a4paper]{article}

\usepackage{amssymb}
\usepackage{amsmath}
\usepackage[dvipsnames]{xcolor}
\usepackage{float}
\usepackage{latexsym}
\usepackage{amsthm}
\usepackage{empheq}
\usepackage{bm}
\usepackage{booktabs}
\usepackage{subcaption}
\usepackage{tikz,pgfplots,cite}
\usepackage{balance}

\usepackage{cite}
\usepackage{amsmath,amssymb,amsfonts}
\usepackage{algorithmic}
\usepackage{graphicx}

\usepackage{latexsym}
\usepackage{amsthm}
\usepackage{thm-restate}
\usepackage{empheq}
\usepackage{bm}
\usepackage{booktabs}
\usepackage{subcaption}
\usepackage{enumitem}
\usepackage[table]{xcolor}
\usepackage{array}
\usepackage[margin=2.8cm]{geometry}

\definecolor{redish}{rgb}{0.9, 0.17, 0.31}
\definecolor{fuchs}{rgb}{0.57, 0.36, 0.51}
\usepackage[colorlinks=true,linkcolor=fuchs]{hyperref}

\theoremstyle{definition}
\newtheorem{theorem}{Theorem}
\newtheorem{corollary}[theorem]{Corollary}
\newtheorem{proposition}[theorem]{Proposition}
\newtheorem{conjecture}{Conjecture}
\newtheorem{definition}[theorem]{Definition}
\newtheorem{example}[theorem]{Example}
\newtheorem{notation}[theorem]{Notation}
\newtheorem{remark}[theorem]{Remark}

\newtheorem{lemma}[theorem]{Lemma}

\usepackage{comment}

\newtheorem{problem}{Problem}

\newcommand{\bfe}{{\boldsymbol e}}

\newcommand{\bfg}{{\boldsymbol g}}
\newcommand{\bfh}{{\boldsymbol h}}

\newcommand{\bfu}{{\boldsymbol u}}
\newcommand{\bfv}{{\boldsymbol v}}

\newcommand{\bfx}{{\boldsymbol x}}

\newcommand{\F}{\mathbb{F}}
\newcommand{\N}{\mathbb{N}}

\newcommand{\Z}{\mathbb{Z}}

\newcommand{\mC}{\mathcal{C}}

\newcommand{\supp}{\operatorname{supp}}

\newcommand\jon[1]{{{\textcolor{orange}{Jonathan: #1}}}}

\def\BibTeX{{\rm B\kern-.05em{\sc i\kern-.025em b}\kern-.08em
    T\kern-.1667em\lower.7ex\hbox{E}\kern-.125emX}}

\newcommand{\bcomment}[1]{{\color{blue}#1}}
\newcommand{\pcomment}[1]{{\color{purple}#1}}

\newcommand{\ey}[1]{{\footnotesize [\pcomment{#1}\;\;\bcomment{--Eitan}]}}

\title{\textbf{Serving Every Symbol: \\ All-Symbol PIR and Batch Codes}}
\usepackage{authblk}

\author{Avital Boruchovsky$^1$, Anina Gruica$^2$, Jonathan Niemann$^2$, and Eitan Yaakobi$^*$}
\affil{$^1$Technion -- Israel Institute of Technology, Israel \\
$^2$Technical University of Denmark, Denmark \\
\vspace{0.5cm}
\small{\texttt{avital.bor@campus.technion.ac.il, \{anigr, jtni\}@dtu.dk, yaakobi@cs.technion.ac.il}}}

\begin{document}
\maketitle

\begin{abstract}
A $t$-all-symbol PIR code and a $t$-all-symbol batch code of dimension $k$ consist of $n$ servers storing linear combinations of $k$ linearly independent information symbols with the following recovery property: any symbol stored by a server can be recovered from~$t$ pairwise disjoint subsets of servers. In the batch setting, we further require that any multiset of size~$t$ of stored symbols can be recovered from~$t$ disjoint subsets of servers. This framework unifies and extends several well-known code families, including one-step majority-logic decodable codes, (functional) PIR codes, and (functional) batch codes.

In this paper, we determine the minimum code length for some small values of $k$ and $t$, characterize structural properties of codes attaining this optimum, and derive bounds that show the trade-offs between length, dimension, minimum distance, and $t$. In addition, we study MDS codes and the simplex code, demonstrating how these classical families fit within our framework, and establish new cases of an open conjecture from~\cite{YAAKOBI2020} concerning the minimal $t$ for which the simplex code is a $t$-functional batch code.
\end{abstract}
\renewcommand{\baselinestretch}{1}\normalsize
\section{Introduction}

\emph{Batch codes} were first introduced by Ishai et al.~\cite{IKOS04}, motivated by load-balancing applications in distributed storage and cryptographic protocols. In their most general form, for integers $1 \le k \le n$, batch codes encode $k$ information symbols into $n$ strings, referred to as buckets. Each bucket contains linear combinations of the information symbols. In this setting, a single user seeks to retrieve a batch of $t$, where $1\le t \le k$, distinct information symbols by reading at most $r$, where $1\le r \le n$, symbols from any given bucket. The primary objective is to minimize the total length of all buckets (the storage overhead) for fixed parameters $k, t, r,$ and $n$.

Ishai et al.~\cite{IKOS04} also proposed a stronger variant known as multiset batch codes. Designed for multi-user settings, this model involves $t$ distinct users, each requesting a specific data item. Since requests may overlap, the total demand constitutes a multiset of the $k$ information symbols. The defining constraint is that each bucket can be accessed by at most one user. A significant special case arises when each bucket contains exactly one symbol. This model is called a \textit{primitive multiset batch code}~\cite{IKOS04} (or simply a \emph{$t$-batch code}) and it is the most studied in the literature.  This model admits a natural algebraic interpretation: $k$ information symbols are encoded into $n$ encoded symbols using a generator matrix $G \in \mathbb{F}_q^{k \times n}$. The matrix~$G$ generates a $t$-batch code if, for every multiset of $t$ requested information symbols $i_1, i_2, \dots, i_t \in [k]$, there exist $t$ pairwise disjoint subsets $R_1, R_2, \dots, R_t \subseteq [n]$ such that the columns of $G$ that are indexed by $R_j$ span the unit vector corresponding to the symbol $i_j$. Throughout this paper, we restrict the term ``batch code'' to refer exclusively to primitive multiset batch codes.

Over the years, numerous works have investigated various extensions and refinements of batch codes. One particularly influential variant is the class of \emph{private information retrieval} (\textit{PIR}) \textit{codes}, introduced in~\cite{FAZELI2015,fazeli2015codes,vardy2023private} as a means to reduce the storage overhead of PIR schemes while maintaining both privacy and low communication complexity. PIR codes can be viewed as a specialized form of batch codes in which each information symbol is required to possess~$t$ mutually disjoint recovery sets. This corresponds to the batch setting in which the $t$ queries are identical, i.e., $i_1 = i_2 = \dots = i_t$.

A further generalization relevant to our work is that of \emph{functional batch codes}, introduced in~\cite{YAAKOBI2020} and later expanded in~\cite{Lev2022}. In this model, the $t$ simultaneous requests may be arbitrary linear combinations of the information symbols rather than individual symbols themselves.

Several additional variants of batch codes have been studied in the literature, though they are less directly related to the focus of this paper. One such variant is that of \emph{combinatorial batch codes}, in which each bucket stores only uncoded copies of the information symbols. These codes have been extensively analyzed in works such as~\cite{stinson2009combinatorial,BSB2011,silberstein2016optimal,shangguan2020sparse,chee2020lower}. A special case with $t = n$, known as \emph{switch codes}, has been explored in~\cite{wang2013codes,wang2015optimal,chee2015combinatorial,buzaglo2017consecutive} in the context of data routing in network switches. More recently,~\cite{kong2024bounds} introduced a related notion called an \emph{$(s,t)$-batch code}, which requires that the multiset of $t$ requested items contain at most $s$ distinct information symbols.

In this paper, we introduce and study generalized versions of batch codes and PIR codes, which we refer to as \emph{all-symbol batch} and \emph{all-symbol PIR} codes. In the all-symbol PIR setting, the goal is to retrieve the same code symbol $t$ times (where the symbol does not need to be an information symbol) using $t$ mutually disjoint recovery sets. In the all-symbol batch setting, the requirement is stronger: for every multiset of $t$ requested code symbols, there must exist $t$ pairwise disjoint recovery sets, one for each requested symbol. These notions extend the traditional PIR and batch frameworks by demanding recoverability not only for information symbols but for all codeword symbols. These definitions unify and generalize several previously studied code properties, including one-step majority-logic decodable codes. Beyond their theoretical interest, these codes are motivated by applications in distributed storage and private information retrieval, where efficient and reliable access to multiple (potentially repeated) codeword symbols is essential.

While all-symbol batch codes have not been studied previously to the best of our knowledge, the notion of all-symbol PIR codes intersects with several previously proposed definitions that arise under specific parameter choices. For example, the $t$ disjoint-repair-group property for $s$ symbols, denoted $(t,s)$-DGRP (Definition 1 in~\cite{KARINGULA_VARDY_WOOTTERS2022}), coincides with the $(t+1)$-all-symbol PIR property when $s = n$. In addition, a one-step majority-logic decodable code with $t$ orthogonal repair sets is precisely a $(t+1)$-all-symbol PIR code (see Chapter 8 in \cite{lin2001error}). These connections are discussed further in Section~\ref{Sec:RelatedWork} and are revisited throughout the paper. 

We focus on two main problems. The first problem is to determine the minimum length of an all-symbol batch/PIR code for given $t$ and $k$. We obtain partial answers to this question for small values of $t$, and discuss general bounds. The second problem is to determine how the parameters of a code influence its potential recovery properties. In particular, we will consider the role that the dual minimum distance plays, and discuss what happens for MDS codes and the simplex code.

The rest of the paper is organized as follows. In 
Section~\ref{Sec:ProbelmStatement}, we define the problems 
studied in this work and establish the necessary notation 
and background. Section~\ref{Sec:3_The_Length} presents 
basic properties of all-symbol PIR and batch codes, 
along with the minimum length required for these codes 
under fixed parameters. In Section~\ref{Sec:Propetries_Of_Known_Codes}, 
we investigate the all-symbol PIR and batch properties of 
several well-known classes of codes. Finally, 
Section~\ref{Sec:Conclustion} provides concluding remarks.

\section{Problem Statement}\label{Sec:ProbelmStatement}

\subsection{Preliminaries and Notation}

Throughout this paper, $k$ and $n$ are integers with $1 \le k \le n$, $q$ is a prime power, $\F_q$ denotes the finite field with $q$ elements, and $t \ge 1$ is an integer. We denote by $[n]$ the set $\{1,\dots,n\}$. For a matrix $G \in \F_q^{k \times n}$, we denote its columns by $\bfg_1, \dots, \bfg_n$. Given a set of vectors $V$, we denote by $\langle V\rangle$ their $\F_q$-span. We let $\bfe_i$ denote the $i$-th unit vector and $\bf1$ the all-one vector, where the dimensions are determined by the context. Finally, for a vector $\bfv$, we denote by~$\bfv^t$ the multiset obtained by repeating $\bfv$ $t$ times. 



In order to introduce the problem this paper focuses on, we need to define what it means for a matrix to \emph{serve} a list of vectors.

\begin{definition}
    Let $G \in \F_q^{k \times n}$ be a matrix and let $\bfv \in \F_q^k$. A set $R \subseteq [n]$ is a \textbf{recovery set} for~$\bfv$ if $v \in \langle \bfg_j : j \in R\rangle$. For a multiset $L:=\{\bfv_1,\dots,\bfv_t\} \subseteq \F_q^k$, we say that $G$ can \textbf{serve} this list, if there exist pairwise disjoint recovery sets $R_1,\dots,R_t \subseteq [n]$ with the property that $\bfv_i \in \langle \bfg_j : j \in R_i \rangle$.
\end{definition}

We are interested in generator matrices of linear codes that can serve special types of lists of vectors. More precisely, we are interested in the following cases.

\begin{definition}
An $\F_q$-linear code $\mC$ in $\F_q^n$ of dimension $k$ is 
\begin{itemize}
    \item[(i)] a \textbf{$t$-PIR} (\textbf{P}) code if there exists a generator matrix $G \in \F_q^{k \times n}$ of $\mC$ that can serve the list $L=\{\bfe_i^t\}$ for all $i \in [k]$; 
    \item[(ii)] a \textbf{$t$-batch} (\textbf{B}) code if there exists a generator matrix $G \in \F_q^{k \times n}$ of $\mC$ that can serve any list $L=\{\bfe_1^{t_1},\dots,\bfe_k^{t_k}\}$ with $t_1+\dots + t_k=t$;
    \item[(iii)] a \textbf{$t$-functional PIR} (\textbf{FP}) code if there exists a generator matrix $G \in \F_q^{k \times n}$ of $\mC$ that can serve the list $L=\{\bfv^t\}$ for all $\bfv \in \F_q^k$;
    \item[(iv)] a \textbf{$t$-functional batch} (\textbf{FB}) code if there exists a generator matrix $G \in \F_q^{k \times n}$ of $\mC$ that can serve the list $L=\{\bfv_1,\dots,\bfv_t\}$ for all $\bfv_1,\dots,\bfv_t \in \F_q$;
    \item[(v)] a \textbf{$t$-all-symbol PIR} (\textbf{ASP}) code if there exists a generator matrix $G \in \F_q^{k \times n}$ of $\mC$ that can serve the list $L=\{\bfg_i^t\}$ for all $i \in [n]$; 
    \item[(vi)] a \textbf{$t$-all-symbol batch} (\textbf{ASB}) code if there exists a generator matrix $G \in \F_q^{k \times n}$ of $\mC$ that can serve any list $L=\{\bfg_1^{t_1},\dots,\bfg_n^{t_n}\}$ with $t_1+\dots + t_n=t$.
\end{itemize}
\end{definition}

We say that a matrix satisfies a given property if it can be used to prove that the corresponding code meets one of the above definitions. For example, a full-rank matrix $G \in \F_q^{k \times n}$ that can serve, for every $i \in [n]$, the list $L=\{\bfg_i^t\}$ is said to satisfy the $t$-all-symbol PIR property. 

In the next lemma, we show that if one generator matrix of a code can serve every list of size $t$ formed from its columns, then this property holds for any generator matrix of the same code. However, the specific lists of vectors that can be served may differ, since they depend on the actual columns of the chosen generator matrix. 

\begin{lemma}\label{Invariance_Of_Mul_By_Invertible}
    Let $G \in \F_q^{k \times n}$, and let $M \in \F_q^{k \times k}$ be an invertible matrix. Then $G$ satisfies the $t$-all-symbol PIR/batch property if and only if $MG$ satisfies the same property.
\end{lemma}
\begin{proof}
   We only prove the case of the $t$-all-symbol batch property, the PIR case can be proven analogously. Let $\{\bfv_{1},\bfv_{2},\ldots,\bfv_{t}\}$ be a multiset of columns of $MG$. Set $\bfu_i=M^{-1}\bfv_i$, then $\bfu_i$ is a column of $G$. As $G$ satisfies the $t$-all-symbol batch property, there exist $t$ pairwise disjoint subsets $R_1,R_2,\ldots,R_t$ of $[n]$ such that $\bfu_i\in \langle \bfg_j : j \in R_i\rangle$ for each $i\in[t]$. Multiplying by $M$ gives $\bfv_i\in \langle M\bfg_j : j \in R_i \rangle$ for each $i\in[t]$. Thus, the matrix $MG$ (that has as columns $\{M\bfg_1,\dots,M\bfg_n\}$) satisfies the $t$-all-symbol batch property as well.
\end{proof}

The preceding lemma establishes that the properties of being $t$-all-symbol PIR or $t$-all-symbol batch are indeed code properties. To facilitate the study of minimum code lengths, we introduce the following notation.

\begin{notation} \label{not:mins}
Let $k, t \in \N$ and $q$ be a prime power. We define the 
optimal lengths for the various code types as follows:
\begin{align*}
    P(k,t,q)&:=\min\{n\in \N : \exists \; k\text{-dim. }t\text{-P code in $\F_q^n$}\}, \\
    B(k,t,q)&:=\min\{n\in \N : \exists \;  k\text{-dim. }t\text{-B code in $\F_q^n$}\}, \\
    FP(k,t,q)&:=\min\{n\in \N : \exists \;  k\text{-dim. }t\text{-FP code in $\F_q^n$}\}, \\
    FB(k,t,q)&:=\min\{n\in \N : \exists \; k\text{-dim. }t\text{-FB code in $\F_q^n$}\}, \\
    ASP(k,t,q)&:=\min\{n\in \N : \exists \;  k\text{-dim. }t\text{-ASP code in $\F_q^n$}\}, \\
    ASB(k,t,q)&:=\min\{n\in \N : \exists \;  k\text{-dim. }t\text{-ASB code in $\F_q^n$}\}.
\end{align*}
\end{notation}

One of our goals is to study $ASP(k,t,q)$ and $ASB(k,t,q)$, and to relate them to $P(k,t,q)$, $B(k,t,q)$, $FP(k,t,q)$ and $FB(k,t,q)$. In the sequel, we say that a matrix $G \in \F_q^{k \times n}$ \emph{realizes} $ASP(k,t,q)$, if $n=ASP(k,t,q)$ and $G$ satisfies the $t$-all-symbol PIR property. This terminology 
is applied analogously to the other properties defined above.

\subsection{Previous Work}\label{Sec:RelatedWork}
Throughout the years, the study of PIR and batch codes, along with their 
generalizations, has garnered significant interest. Given the diversity of notations 
employed across the literature, we provide here a 
unified summary of the foundational results that serve 
as the basis for our work.

The concept of PIR codes was introduced in 
\cite{fazeli2015codes}, with a more comprehensive 
treatment in \cite{FAZELI2015} and a final journal 
version in \cite{vardy2023private}. These works 
primarily characterize the asymptotic behavior of 
PIR codes, establishing that $\lim_{t \to \infty} 
\frac{P(k,t,2)}{t} = 1$. While such asymptotics 
are outside the scope of this paper, these foundational 
works also established subadditivity 
properties for $P(k,t,q)$ and several general bounds 
that we will apply in our derivations (we will explicitly identify and cite 
these bounds as they are applied).

A central property of $t$-PIR codes is that they must 
possess a minimum distance of at least $t$ (see, e.g., 
\cite{Skachek2018}). Consequently, the Singleton 
bound provides a universal lower bound on the 
optimal length:
\begin{equation}
    t + k - 1 \leq P(k,t,q). \label{eq:min_dis_atleast_t}
\end{equation}

For the specific case of dimension $k=2$, the lower 
bound $P(k,t,2) \geq {(2^k-1)t}/{2^{k-1}}$
(see, e.g., \cite[Theorem 9]{vardy2023private}) meets the exact value 
for functional batch codes established in 
\cite[Corollary 3.5]{KILIC2025}, leading to the 
following characterization.

\begin{lemma}\label{lem:k=2,q=2}
    For $k=2$ and $q=2$, it holds that:
    \[
    P(2,t,2) = B(2,t,2) = FP(2,t,2) = FB(2,t,2) 
    = t + \left\lceil \frac{t}{2}\right\rceil.
    \]
\end{lemma}

The distinctions between different code classes often 
diminish for small values of $t$. Specifically, for 
$t=3$, the requirements for batch and PIR codes 
coincide (see \cite[Lemmas 3 and 4]{VARDY_YAAKOBI_2016}). Inspired by this, and following similar ideas, we prove a more general result in Lemma~\ref{lemma:t=3_pir_iff_batch}.

The optimal length of a 3-all-symbol batch code
is given by the following combinatorial expression
\cite{FAZELI2015, fazeli2015codes, 
KARINGULA_VARDY_WOOTTERS2022, Vardy2016}:

\begin{lemma}\label{lem:B=P=k+r_For_t=3}
    For $t=3$, $P(k,3,q) = B(k,3,q) = k + r$, 
    where $r = \min\{i \in \mathbb{N} : \binom{i}{2} \geq k\}$.
\end{lemma}

Expanding upon this, it was shown in 
\cite{vardy2023private} and \cite{FAZELI2015} that 
appending a parity column to the $t=3$ construction 
yields an optimal $t=4$ binary code:

\begin{lemma}\label{lem:B=P=k+r_For_t=4}
   In the binary case $q=2$, we have: 
   $B(k,4,2) = P(k,4,2) = P(k,3,2) + 1$.
\end{lemma}

The concepts of functional PIR and batch codes were 
introduced in \cite{YAAKOBI2020}, sparking 
considerable subsequent research into their 
optimal lengths. Regarding the functional PIR case 
for $t=3$, the following bounds and exact values 
were established in \cite[Corollary 16]{YAAKOBI2020}: 
\begin{lemma}\label{lem:FP_for_t=3}
 For any $m \geq 2$, we have that  \begin{align*}
     &FP(2m,3,2) = 3m + 2,\\
     3m + 3 \leq &FP(2m + 1,3,2) \leq 3m + 4.
 \end{align*}  
\end{lemma}

Several works have also explored properties related to 
all-symbol PIR codes. For instance, a one-step majority-logic 
decodable code with $t$ orthogonal repair sets 
corresponds to a $(t+1)$-all-symbol PIR code (see 
Chapter 8 in~\cite{lin2001error}). That work presents several classes of cyclic majority-logic decodable codes.  Additionally, the 
$(t,s)$-disjoint-repair-group property (DGRP), 
introduced in \cite{KARINGULA_VARDY_WOOTTERS2022}, 
coincides with the $(t+1)$-all-symbol PIR requirement when $s=n$. Moreover, if we impose a restriction on the size of the recovery sets, namely that they are of size at most $r$, then such a constrained $t$-all-symbol PIR code corresponds to a code with locality $r$ and availability $t$~\cite{RAWAT2016}.

Using the bounds and optimal constructions 
for $2$-DGRP codes with $s=n$ provided in 
\cite[Theorem 1, Example 1]{KARINGULA_VARDY_WOOTTERS2022}, 
we obtain:

 \begin{lemma}\label{lem:ASP_For_t=3}
    For $t=3$, it holds that:
    \[
     ASP(k,3,q) = k + r,
    \]
    where $r$ is the smallest integer such that 
    $\binom{r}{2} \geq k$. 
 \end{lemma}
Notably, the construction that achieves the optimal  length for $ASP(k,3,q)$ is the same as the one 
employed for PIR codes in Lemma~\ref{lem:B=P=k+r_For_t=3}. 
We provide a detailed description of this construction in Section~\ref{subsec:t3}, where we utilize 
it to derive further results.

\subsection{Our Contribution}

In this paper, we focus on two basic problems regarding codes satisfying the $t$-all-symbol PIR and the $t$-all-symbol batch property. 

\begin{problem} \label{prob:1}
    For fixed $t$ and $k$, what is the smallest $n$ such that there exists a code satisfying the $t$-all-symbol PIR, and the $t$-all-symbol batch property, respectively?
\end{problem}

Consider the scenario where our code is $2$-dimensional over $\mathbb{F}_2$, and we want to be able to serve any list of $2$ vectors formed from its columns. The most straightforward construction of such a code is the parity code with generator matrix
\begin{align*}
G=\begin{pmatrix}
1 & 0 & 1 \\
0 & 1 & 1
\end{pmatrix} \in \F_2^{2 \times 3},\end{align*}
for which it is clear that any list of size $2$ made from the columns of $G$ can be served. In fact, this is the shortest $2$-dimensional code satisfying the $2$-all-symbol batch property.

However, the situation becomes more complicated when considering codes of larger dimension or larger $t$. In the first part of the paper, we derive closed formulas for small values of the dimension and of $t$, and obtain some general bounds.

In the second part of the paper, we consider codes with fixed parameters (such as length, dimension, and minimum distance) and investigate how well a code with those parameters can perform with respect to being $t$-all-symbol PIR or $t$-all-symbol batch.  That is, we determine bounds on the value of $t$ for which these properties can hold.

\begin{problem} \label{prob:2}
    For a fixed code $\mC$, what is the largest $t$ for which this code has the $t$-all-symbol PIR, and the $t$-all-symbol batch property, respectively? 
\end{problem}

Finally, we consider two famous families of codes (MDS and simplex codes) and analyze how well they perform relative to the bounds previously derived. This analysis reveals a clear connection between codes with the $t$-all-symbol batch property and an open conjecture from 2020~\cite{YAAKOBI2020} concerning the simplex code.

\section{The Length of ASP and ASB Codes}\label{Sec:3_The_Length}

In this section, we focus on Problem~\ref{prob:1}. 

\subsection{Basic Properties}

We begin with some preliminary results and observations regarding all-symbol PIR and all-symbol batch codes. Because of Lemma~\ref{Invariance_Of_Mul_By_Invertible} the choice of the generator matrix does not matter, and so in the sequel, we mainly focus on systematic generator matrices.

The following are some straightforward results for $ASP(k,t,q)$ and $ASB(k,t,q)$, in relation with the other values introduced in Notation~\ref{not:mins}.

\begin{proposition}\label{prop:first}
    We have that
    
    \begin{itemize}
    \item[(i)] 
    \( P(k,t,q) \le ASP(k,t,q) \le FP(k,t,q) \),
    
    \( B(k,t,q) \le ASB(k,t,q) \le FB(k,t,q) \).

    \item[(ii)] 
    \( ASP(k,t,q) \le ASB(k,t,q) \).

    \item[(iii)] \text{Strict monotonicity in \(t\)}:  
    \begin{align*}
        &ASP(k,t-1,q) \le ASP(k,t,q) - 1, \\
        &ASB(k,t-1,q) \le ASB(k,t,q) - 1.
    \end{align*}

    \item[(iv)] \text{Subadditivity in \(k\)}:  
    \begin{align*}
        &ASP(k_1 + k_2, t, q) \le ASP(k_1, t, q) + ASP(k_2, t, q), \\
        &ASB(k_1 + k_2, t, q) \le ASB(k_1, t, q) + ASB(k_2, t, q).
    \end{align*}
\end{itemize}
\end{proposition}

\begin{proof} \text{ }
    \begin{itemize}
         \item[(i)+(ii)] 
    These inequalities follow directly from the definitions.
          \item[(iii)] Let $G\in \F_q^{k \times n}$ be a matrix that realizes $ASB(k,t,q)$. Deleting any column of $G$ yields a $(t-1)$-all-symbol batch code; see \cite[Theorem 2]{YAAKOBI2020}. The same argument applies to $ASP$.
          \item[(iv)] Let $A\in \F_q^{k_1 \times n}$ and $B\in \F_q^{k_2 \times n}$ be matrices that realize  $ASB(k_1,t,q)$ and $ASB(k_2,t,q)$, respectively. Then the block-diagonal matrix $ 
          \begin{bmatrix}
          A & 0\\
          0 & B
          \end{bmatrix}$
          satisfies the \(t\)-all-symbol batch property, establishing the subadditivity. The same argument applies to $ASP$.  \qedhere
    \end{itemize}
\end{proof}

While the subadditivity in $k$ is easy to see, it is less clear whether subadditivity in $t$ holds as well, as it was shown for $FP(k,t,q)$ and $FB(k,t,q)$; see \cite[Proposition 2.6]{KILIC2025}. Nevertheless, we have the following (partial) result.

\begin{lemma}
    For any $\lambda\in \mathbb{N}$ we have 
    $$
        ASP(k,\lambda t,q)\leq  \lambda ASP(k,t,q), \quad ASB(k,\lambda t,q)\leq  \lambda ASB(k,t,q).
    $$
\end{lemma}
\begin{proof}
    By horizontally joining $\lambda$ copies of a matrix realizing $ASB(k,t,q)$, we obtain a matrix satisfying the $\lambda t$-all-symbol batch condition. Similarly for $ASP$.
\end{proof}

We conjecture that the subadditivity in $t$ holds in general; this remains an open direction for future work.
\begin{conjecture}
    There is subadditivity in $t$, i.e., for all $k,q$ and $t_1,t_2\geq1$, we have 
    \begin{align*}
        ASP(k,t_1+t_2,q)&\leq ASP(k,t_1,q)+ASP(k,t_2,q), \\
        ASB(k,t_1+t_2,q)&\leq ASB(k,t_1,q)+ASB(k,t_2,q).
    \end{align*}
\end{conjecture}

Next, we determine $ASB(k,t,q)$ and $ASP(k,t,q)$ for some small values of $k$ and $t$.

\begin{proposition}\label{prop:k=1_t=1_t=2}
    We have that
    \begin{itemize}
        \item[(i)] $ASP(1,t,q) = ASB(1,t,q)=t$,
        \item [(ii)] $ASP(k,1,q) = ASB(k,1,q) =k$, and 
        \item [(iii)] $ASP(k,2,q) = ASB(k,2,q) = k + 1$.
    \end{itemize}
\end{proposition}

\begin{proof}
    \text{ }
    \begin{itemize}

        \item[(i)] For $k=1$, serving $t$ requests requires $t$ disjoint recovery sets. Hence at least $t$ columns; the matrix $(1,\dots,1) \in \F_q^{1 \times t}$ satisfies this.
        
        \item[(ii)] Any matrix realizing $ASB(k,1,q)$ must have rank $k$, and therefore must contain at least $k$ columns.  
    Equality is achieved by the $k \times k$ identity matrix.
        
        \item[(iii)] By Lemma~\ref{Invariance_Of_Mul_By_Invertible} we can assume that the matrix $G\in \F_q^{k \times n}$ realizing $ASB(k,2,q)$ is systematic. If $n=k$, then it is not possible to serve a request of the form $\{\bfe_i,\bfe_i\}$, hence $n\geq k+1$. Equality is achieved by taking the identity matrix with a global parity column. \qedhere
    \end{itemize}
\end{proof}

The following is a general lower and upper bound on $ASP(k,t,q)$ and $ASB(k,t,q)$.

\begin{proposition} \label{prop:lbub}
    We have that
    \begin{align*}
        \max(t + k - 1,\left\lceil \frac{2(k+1)t}{k+2} \right\rceil) \le ASP(k,t,q) \le ASB(k,t,q) \le  \left\lceil \frac{(k+1)t}{2} \right\rceil.
    \end{align*}
\end{proposition}
\begin{proof}

    For $t = 1$ and $t = 2$ the statement follows from Proposition \ref{prop:k=1_t=1_t=2}, so suppose $t \geq 3$. 
    
    First, by Equation~\eqref{eq:min_dis_atleast_t} and Proposition~\ref{prop:first}, we have that $t + k - 1\le P(k,t,q)\le ASP(k,t,q)$. We next show that $ASP(k,t,q) \geq {2(k+1)t}/{(k+2)}$. Let $\mathcal{C}\subseteq \F_q^n$ be a $t$-all-symbol PIR code of dimension $k$ and length $n = ASP(k,t,q)$, and let $G$ be any generator matrix of $\mathcal{C}$. Suppose that $G$ has $m \leq n$ distinct columns $\bfg_1, \dots, \bfg_{m}$ appearing $n_1, \dots, n_{m}$ many times, respectively. Note that then $n_1 + \cdots + n_m = ASP(k,t,q)$. 
    
    Now, fix some $j \in [m]$. If $t < n_j$ then $G$ cannot attain $ASP(k,t,q)$, since the extra $n_j-t$ columns can be removed and the respective generated code is still $t$-all-symbol PIR, so we have $t \geq n_j$. By assumption $G$ can serve $\{\bfg_j^t\}$, and there can be at most $n_j$ recovery sets of size one. Hence, we have
    $$
        ASP(k,t,q) \geq n_j + 2(t - n_j).
    $$

    Summing over all $j \in [m]$ we obtain
    $$
        m\cdot ASP(k,t,q) \geq ASP(k,t,q)  + 2mt - 2ASP(k,t,q),
    $$
    and rearranging, and using that $ASP(k,t,q)$ is an integer, yields
    $$
        ASP(k,t,q) \geq \left\lceil \frac{2mt}{m+1}\right\rceil.
    $$ 

    Since $G$ has rank $k$ we have $m \geq k$, but we claim that also $m \geq k+1$ holds. In fact, if $m = k$ then we may assume $\bfg_i = \bfe_i$ for $i \in [k]$, and then $n_1 = \dots= n_k = t$ and $n = kt$ is the only option. However, by removing $k$ columns of $G$, one $\bfe_i$ for each $i \in [k]$, and inserting one column equal to $\bfe_1 + \dots + \bfe_k$, we get a matrix that still satisfies the $t$-all-symbol PIR property. This is in contradiction with the assumption that $G$ realizes $ASP(k,t,q)$, so we conclude $m \geq k+1$. 

    Hence, 
    $$
        ASP(k,t,q) \geq \left\lceil \frac{2mt}{m+1}\right\rceil \geq \left\lceil \frac{2(k+1)t}{k+2}\right\rceil
    $$
    as claimed.

    To finish the proof, we show that $ASB(k,t,q) \leq \left\lceil {(k+1)t}/{2} \right\rceil$ by explicitly constructing a generator matrix for a $t$-all-symbol batch code over $\F_q$ of dimension $k$ and length $ \left\lceil {(k+1)t}/{2} \right\rceil$. Let $G$ be the matrix that contains $\lceil\frac{t}{2}\rceil$ copies of $\bfe_i$ for all $i \in [k]$, and $\lfloor\frac{t}{2}\rfloor$ copies of $\bf1$. To verify that $G$ satisfies the $t$-all-symbol batch condition, consider a request (multiset) $L = \{\mathbf \bfe_1^{\,t_1},\, \dots,\, \mathbf \bfe_k^{\,t_k},\, \mathbf 1^{\,t_{k+1}}\},$ with $t_1 + \dots + t_{k} + t_{k+1} = t$. If $t_i\le \lceil\frac{t}{2}\rceil$ for every $i \in [k]$, then the required recovery sets are immediate. Moreover, observe that at least  $k$ of the integers $t_i$, $i \in [k+1]$ are less than or equal to $\lfloor \frac{t}{2}\rfloor$. Thus, we are left with two cases:
    
     \textbf{Case 1:} Suppose that $t_{k+1} > \lfloor \frac{t}{2}\rfloor$ and $t_i \le \lfloor\frac{t} {2}\rfloor$ for all $i \in [k]$. We clearly have enough size-one recovery sets for each $\bfe_i$, $i \in [k]$. Moreover, since $t_1+\dots +t_{k+1}=t=\lceil\frac{t}{2}\rceil+\lfloor\frac{t} {2}\rfloor$ we obtain
    \begin{align*}
        t_{k+1}-\left\lfloor\frac{t}{2}\right\rfloor = \left\lceil\frac{t}{2}\right\rceil-(t_1 + \dots + t_{k}) \le \left\lceil\frac{t}{2}\right\rceil-t_i
    \end{align*}
    for all $i \in [k]$. Thus, for every such \(i\), there remain enough unused columns of type \(\bfe_i\) to construct additional recovery sets for $\bf1$. In total, we may recover $\bf1$ using \(\lfloor t/2 \rfloor\) size-one recovery sets and \(t_{k+1} - \lfloor t/2 \rfloor\) recovery sets formed from the remaining columns.
    
    \textbf{Case 2:} Suppose that $t_1 > \lfloor \frac{t}{2}\rfloor$ and $t_i \le \lfloor\frac{t} {2}\rfloor$ for all $i \in \{2,\dots,k+1\}$. This time, we have enough recovery sets of size one for the last $k$ columns. Moreover, since $t_1+\dots +t_{k+1}=t=\lceil\frac{t}{2}\rceil+\lfloor\frac{t} {2}\rfloor$ we obtain
    \begin{align*}
        t_{1}-\left\lceil\frac{t}{2}\right\rceil = \left\lfloor\frac{t}{2}\right\rfloor-(t_2 + \dots + t_{k+1}) \le \left\lfloor\frac{t}{2}\right\rfloor-t_i
    \end{align*}
    for all $i \in \{2,\dots,k+1\}$. Thus, we may recover $\bfe_1$ using \(\lfloor t/2 \rfloor\) size-one recovery sets and \(t_{k+1} - \lfloor t/2 \rfloor\) recovery sets formed from the remaining columns. Note that we can prove the above statement in an analogous way in the case where $t_i > \lceil\frac{t}{2}\rceil$ for any $i \in \{2,\dots,k\}$. We conclude that $G$ satisfies the $t$-all-symbol batch condition, and this finishes the proof.
\end{proof}

From the proof of Proposition~\ref{prop:lbub} we observe the following refinement of the lower bound.

\begin{corollary}
    Let $\mathcal{C} \subseteq \F_q^n$ be a code of dimension $k$ that satisfies the $t$-all-symbol PIR property. If the generator matrix $G$ of $\mC$ has $\gamma$ distinct columns, then
    $$
        n\geq \left\lceil \frac{2\gamma t}{\gamma+1}\right\rceil.
    $$ In particular, if the generator of $\mathcal{C}$ has all distinct columns, or equivalently $d^\perp \ge 3$, then $2t-1 \le n$.
\end{corollary}

Observe that when $k$ is large compared to $t$, the bound 
$t + k - 1$ is tighter than $\lceil {2(k+1)t}/{(k+2)}\rceil$. 
Conversely, for small values of $k$, the latter bound performs better. In particular, by setting $k=2$ in 
Proposition~\ref{prop:lbub}, the lower and upper bounds 
coincide, yielding the following corollary.

\begin{corollary}\label{prop:k=2}
    It holds that $ASP(2,t,q) = ASB(2,t,q) = t+\left\lceil\frac{t}{2}\right\rceil.$
\end{corollary}

Note that for $q = 2$, the result in 
Corollary~\ref{prop:k=2} follows directly from  Lemma~\ref{lem:k=2,q=2} and 
Proposition~\ref{prop:first}.

In the rest of this subsection we investigate structural properties of matrices realizing $ASP(k,t,q)$ and $ASB(k,t,q)$, respectively. The first question we answer is whether such a matrix can have repeated columns. The answer turns out to be yes in general; we will see in Section~\ref{subsec:t3} that $ASP(4,3,2)=ASB(4,3,2) = 8$, and the matrix
    $$
        G :=
        \begin{pmatrix}
            1 & 0 & 0 & 0 & 0 & 1 & 1 & 1 \\
            0 & 1 & 0 & 0 & 1 & 1 & 1 & 1 \\
            0 & 0 & 1 & 0 & 0 & 0 & 1 & 1 \\
            0 & 0 & 0 & 1 & 1 & 0 & 1 & 1 \\
        \end{pmatrix} \in \F_2^{4 \times 8}
    $$
satisfies the 3-all-symbol batch property. However, in some special cases, a matrix realizing $ASB(k,t,q)$ or $ASP(k,t,q)$, respectively, must have pairwise distinct columns, as we prove in Lemma~\ref{lemma:jump_repeated}. We further note that Lemma~\ref{lemma:jump_repeated} applies equally to functional PIR and functional batch codes. To the best of our knowledge, this observation has not previously appeared in the literature.

\begin{lemma}\label{lemma:jump_repeated}
    If $ASB(k+1,t,q)=ASB(k,t,q)+1$, then any matrix $G\in \F_q^{(k+1) \times n}$ that realizes $ASB(k+1,t,q)$, has pairwise distinct columns. The same holds true if we substitute $ASB$ with $ASP$, $FP$ or $FB$.
\end{lemma}
\begin{proof}
    Let $G$ be a matrix that realizes $ASB(k+1,t,q)$. Suppose, towards a contradiction, that~$G$ has a column $\bfg$ appearing at least twice, and without loss of generality assume that it appears as the first two columns of $G$. By Lemma~\ref{Invariance_Of_Mul_By_Invertible}, we may assume that $G$ has the form
    \[
        G=
        \begin{pmatrix}
        1 & 1 & *\\
        \bf0 & \bf0 & B
        \end{pmatrix},
        \qquad
        B\in\mathbb{F}_q^{k\times (n-2)}.
    \]
    But then $B$ satisfies the $t$-all-symbol batch property, contradicting the assumption that $ASB(k+1,t,q)=ASB(k,t,q)+1$. The same argument applies to $ASP$, $FP$ and $FB$.
\end{proof}

\subsection{The Case $t=3$} \label{subsec:t3}

It is known that
$$
    P(k,3,q) = B(k,3,q) = k + r,
$$
where $r$ is the smallest integer such that $\binom{r}{2} \geq k$ (Lemma \ref{lem:B=P=k+r_For_t=3}). Using similar ideas, and a similar construction, it can be shown that this is also the case for 3-all-symbol batch and 3-all-symbol PIR codes.

\begin{lemma}\label{lemma:t=3_pir_iff_batch}
    A code $\mathcal{C}\subseteq \F_q^n$ satisfies the $3$-all-symbol batch property if and only if it satisfies the $3$-all-symbol PIR property. In particular, $ASB(k,3,q) = ASP(k,3,q)$.
\end{lemma}

\begin{proof}
    Let $G$ be a generator matrix of a $3$-all-symbol PIR code. To show that $G$ also satisfies the $3$-all-symbol batch property, it suffices to verify that any request containing two copies of a column and one copy of another can be served. Without loss of generality, consider the multiset \(\{\bfg_1, \bfg_1, \bfg_2\}\). By the PIR property, the column \(\bfg_1\) has three pairwise disjoint recovery sets. At most one of these sets can contain \(\bfg_2\). Hence, at least two of the recovery sets of \(\bfg_1\) avoid using \(\bfg_2\), and these can be used to recover the two copies of \(\bfg_1\).
The remaining singleton set \(\{\bfg_2\}\) serves as a recovery set for \(\bfg_2\).
\end{proof}

By combining Lemma~\ref{lem:B=P=k+r_For_t=3} and Lemma~\ref{lem:ASP_For_t=3}, we obtain the exact values of $ASP(k,3,q)$ and $ASB(k,3,q)$, as presented in the following proposition. For convenience, and since we will reference to this construction later in this work, we provide the proof below.
\begin{proposition}\label{prop:t=3}
    We have
    $$
        ASP(k,3,q) = ASB(k,3,q) = k + r,
    $$
    where $r$ is the smallest integer such that $\binom{r}{2} \geq k$. Equivalently, it holds that
    $$
        ASP(k,3,q) = ASB(k,3,q) = k + \left\lceil \frac{1 + \sqrt{1 + 8k}}{2}\right\rceil.
    $$
\end{proposition}

\begin{proof}
    Let $k$ and $q$ be given, and let $r$ be as above. It follows from $P(k,3,q) = k+r$ that $ASP(k,3,q)$ and $ASB(k,3,q)$ cannot be smaller than $k+r$. Hence, we are done if we can construct a $k\times (k+r)$ matrix over $\F_q$, which satisfies the $3$-all-symbol batch property. To do so, note that we can choose $k$ distinct vectors of weight 2 in $\F_2^{r}$. These vectors can be considered as elements of $\F_q^{r}$ instead, and we let $A\in \F_q^{k\times r}$ denote the matrix which has them as its rows. Now it is not to hard to show that $G = (I_k \mid A)$ satisfies the $3$-all-symbol PIR property, and hence, by Lemma \ref{lemma:t=3_pir_iff_batch}, also the $3$-all-symbol batch property. 
\end{proof}

\begin{remark}
    To the best of our knowledge, the exact value of 
    $FB(k,3,q)$ remains unknown; for asymptotic bounds, 
    we refer the reader to \cite[Table~IV]{FAZELI2015} 
    and \cite{KILIC2025}. This stands in contrast to 
    functional PIR codes, as characterized in 
    Lemma~\ref{lem:FP_for_t=3}. By comparing 
    Proposition~\ref{prop:t=3} with 
    Lemma~\ref{lem:FP_for_t=3}, we observe that 
    constructing a functional PIR code requires 
    approximately $O(k-\sqrt{k})$ additional columns. 
    More precisely, relative to the all-symbol 
    quantities in Proposition~\ref{prop:t=3}, a 
    functional PIR code requires approximately 
    $\frac{k}{2} - \sqrt{2k}$ extra columns, 
    within a tolerance of $\pm 4$ columns.
\end{remark}

\subsection{The Case $t=4$}

Recall that $B(k,4,2) = P(k,4,2) = P(k,3,2) + 1$, as 
stated in Lemma~\ref{lem:B=P=k+r_For_t=4}. By 
extending the techniques used there to the 
all-symbol case, together with the results from Section~\ref{subsec:t3}, we provide bounds for $ASP(k,4,q)$ and $ASB(k,4,q)$. Notably, our proof generalizes the previously 
mentioned binary results to arbitrary $q$ (see 
Corollary~\ref{COR:BatchFor_k=4_general_q}) and 
yields exact values for even $q$ and certain values of $k$ (see Proposition \ref{prop:t=4_q=2}).

We next define the matrices that will be used to derive the results of this section. Let $G = (I_k \mid A)$ be the matrix constructed in the proof of Proposition \ref{prop:t=3}, except that all entries equal to $1$ in $A$ are replaced by $-1$. Note that $G \in \F_q^{k \times n}$ where $$n = k + r = k + \left\lceil \frac{1 + \sqrt{1 + 8k}}{2}\right\rceil,$$ and that $G$ still has the $3$-all-symbol batch property. 

Now let \(G'\in\F_q^{k\times(n+1)}\) be obtained from \(G\) by adding a single parity-check column, namely, the unique column for which the sum of all columns of \(G'\) is the zero vector. Similarly, let \(G''\in\F_q^{k\times(n+2)}\) be obtained from \(G'\) by adding one more copy of this parity-check column. By the construction of \(A\), this parity-check column is precisely the vector $\bf1$ (see Example~\ref{Ex:t=4}, where \(G'\) is displayed for \(k=5\) and \(G''\) for \(k=6\)).

We next prove a couple of lemmas on the way to the main result of this section.
\begin{lemma}\label{Lem:G'IsBatch}
    The matrix $G'$ can serve any list of four columns from $I_k$.
\end{lemma}

\begin{proof} 
    We will first show that for any column $\bfg$ of $I_k$, the columns of $G'$ can be partitioned into four disjoint sets, each of which forms a recovery set for $\bfg$. So suppose $\bfg=\bfe_{i_0}$ for some $i_0\in[k]$. 
    
    For the first recovery set, take $R_1 := \{i_0\}$. For the second and third recovery sets, the construction of $G$ guarantees the existence of two redundancy columns whose entries in row~$i_0$ are equal to~$1$, and whose remaining support is disjoint. Denote these two columns by $\bfh_2$ and $\bfh_3$. For $R_2$ and $R_3$, we take the indices of these two columns together with the indices corresponding to the supports of $\bfh_{2}$ and $\bfh_{3}$ (excluding $i_0$), so that
    
    $$
         \bfe_{i_0} = - \left( \bfh_{2} + \sum_{i \in supp(\bfh_{2}) \setminus\{i_0\}} \bfe_i \right) = - \left( \bfh_{3} + \sum_{j \in supp(\bfh_{3})\setminus\{i_0\}} \bfe_j \right).
    $$
    
    Let $T$ be the remaining indices of columns of $G'$. Then, 
    \begin{align*}
        0 & =   \sum_{i\in R_1} \bfg_i + \sum_{j\in R_2}\bfg_j + \sum_{l\in R_3} \bfg_l + \sum_{m\in T}\bfg_m \\
        &= \bfg - \bfg - \bfg  + \sum_{m\in T}\bfg_m = - \bfg + \sum_{m\in T}\bfg_m,
    \end{align*}
    so that 
    $$
        \bfg  =  \sum_{m\in T}\bfg_m.
    $$
    We conclude that the columns with indices in $R_4 =  T$ can be taken as a fourth recovery set for $\bfg$, disjoint from the existing recovery sets, and that $[n+1] = R_1 \cup R_2 \cup R_3 \cup R_4$. 

Next, we need to show that $G'$ can also serve lists of four not necessarily equal columns from $I_k$. It is straightforward to verify that any list in which at most one vector appears more than once can already be served by $G'$ (see \cite[Lemma 3]{VARDY_YAAKOBI_2016}). Thus, the only remaining case is when $L = \{\bfe_{i_1}^2, \bfe_{i_2}^2\}$, where ${i_1}$ and ${i_2}$ are distinct elements from $[k]$. By the argument above, we may find recovery sets $R_1^1, R_2^1,R_3^1, R_4^1$ of $\bfe_{i_1}$ and $R_1^2, R_2^2,R_3^2, R_4^2$ of $\bfe_{i_2}$, such that, for each $j = 1,2,3,4$ we have
    \[
    \bfe_{i_1} = \sum_{i \in R_j^1} \alpha_i \bfg_i,
    \qquad
    \bfe_{i_2} = \sum_{i \in R_j^2} \beta_i \bfg_i,
\]
    with $\alpha_i, \beta_i \in \F_q$, and such that the four recovery sets of each vector form a partition: 
    $$
        \bigsqcup_{j=1}^4 R^1_{j} = \bigsqcup_{j=1}^4 R^2_{j} = [n+1].
    $$

    Moreover, without loss of generality we may assume that $\bfe_{i_2} \in R^1_{3}$ and that $R^2_{1} = \{\bfe_{i_2}\}$. Then, 
    \begin{align*}
        \bfe_{i_2} &= \alpha_{i_2}^{-1} \left( \bfe_{i_1} - \sum_{i \in R^1_{3}\setminus\{i_2\}} \alpha_i \bfg_i \right) \\
        &= \alpha_{i_2}^{-1} \left( \sum_{i \in R^1_{4}} \alpha_i\bfg_i - \sum_{i \in R^1_{3}\setminus\{i_2\}} \alpha_i \bfg_i \right),
    \end{align*}
    so we can choose $R^1_{1}, R^1_{2}, R^2_1$ and $(R^1_4 \cup R_3^1) \setminus\{i_2\}$ as our recovery sets. The result follows.
\end{proof}

Note that for PIR and batch codes we consider only unit vectors. Therefore, the above lemma shows that the results from \cite[Lemma 5]{VARDY_YAAKOBI_2016} and \cite[Lemma 14]{FAZELI2015} hold not just for binary codes, but for any $q$:

\begin{corollary}\label{COR:BatchFor_k=4_general_q}
    For every $q$, we have that
    $$
        B(k,4,q) = P(k,4,q) = P(k,3,q) + 1.
    $$
\end{corollary}


We now turn our attention to $G''$.

\begin{lemma}\label{lemma:t=4_G''}
    For odd $q$ and for any $k$, the matrix $G''$ satisfies the 4-all-symbol batch property.
\end{lemma}

\begin{proof}

    We first show that $G''$ satisfies the $4$-all-symbol PIR property. For columns of $I_k$ the previous lemma gives four disjoint recovery sets. For the parity column, the recovery sets are obtained by taking the two copies of the parity column, the columns of we $I_k$, and the columns of $A$ (here we exploit the fact that $q$ is odd). The case for columns of $A$ is slightly more complicated:

    Suppose we want to find four recovery sets for a column $\bfg_{i_0}$ with $k < i_0 \leq n$, i.e., a column of $A$. As a first recovery set we choose $R_1 = \{ i_0\}$. A second recovery set can be found using only columns from $I_k$, in fact one can choose
    $$
        R_2 := \{ i : i \in \supp \bfg_{i_0} \}.
    $$
    For the third recovery set we choose the remaining columns of $I_k$ together with one of the parity columns, i.e., 
    $$
        R_3 := [k]\setminus R_2 \cup \bfg_{n+2}.
    $$
    Finally, we claim that
    $$
        R_4 := \{k+1,  k+2, \dots, n+1\} \setminus \{ i_0 \},
    $$
    i.e., $R_4$ corresponds to all columns of $A$ except $\bfg_{i_0}$ together with a parity column. To see that this is a recovery set note that
    $$
        \sum_{i= k+1}^{n} \bfg_i= -2\cdot \bfg_{n+1},
    $$
    so we have
    $$
        \bfg_{i_0} = -\sum_{i= k+1}^{i_0-1} \bfg_{i}-\sum_{j= i_0+1}^{n} \bfg_{i} - 2\bfg_{n+1}.
    $$
    This shows that $R_4$ is also a recovery set, and we can conclude that $G''$ satisfies the $4$-all-symbol PIR property.

    To go from $4$-all-symbol PIR to $4$-all-symbol batch the only non-trivial thing we need to check is that a list of two copies of two distinct vectors can be served, i.e., that the list $L = \{\bfg_{i_1}^2,\bfg_{i_2}^2\}$ can be served for any $1\leq i_1 < i_2 \leq n+1$ (since $\bfg_{n+1} = \bfg_{n+2}$). If $i_1,i_2 \in [k]$ then this follows from Lemma \ref{Lem:G'IsBatch}. If $i_2 = n+1$, then we can simply choose $\{n+1\}$ and $\{n+2\}$ as recovery sets for $\bfg_{i_2}$, and $\{i_1\}$ and $[k]$ as recovery sets for $\bfg_{i_1}$. In all other cases, we have $k < i_2 < n+1$ and the following choice of recovery sets works:

    For $\bfg_{i_2}$ we choose $\{i_2\}$ as well as $R_4$ from above. Note that $R_4$ might contain $i_1$, so we cannot choose $\{i_1\}$ as a recovery set for $\bfg_{i_1}$ in general. However, the columns we have not used yet are exactly those in $I_k$ together with the all 1 vector $\bfg_{n+1}$, from which we can easily find two disjoint recovery sets for $\bfg_{i_1}$. In fact, the matrix $(I_k \mid \bfg_{n+1})$ even satisfies the $2$-functional PIR property.
    
    We conclude that all lists of $4$ columns from $G''$ can be served, and hence that $G''$ satisfies the $4$-all-symbol batch property.
\end{proof}

We are now ready to prove one of the main results of this section.

\begin{theorem} \label{thm:t=4} 
    For odd $q$, we have
    $$
        k+1 + \left\lceil \frac{1 + \sqrt{1 + 8k}}{2}\right\rceil \leq ASP(k,4,q) \leq ASB(k,4,q) \leq k+2 + \left\lceil \frac{1 + \sqrt{1 + 8k}}{2}\right\rceil.
    $$
\end{theorem}

\begin{proof}
    The lower bound is immediate from Propositions \ref{prop:first} and \ref{prop:t=3}. The upper bound follows from Lemma~\ref{lemma:t=4_G''}, in which we showed that \(G''\) satisfies the \(4\)-all-symbol batch property.
\end{proof}

Despite the fact that Theorem~\ref{thm:t=4} only applies when $q$ is odd, for even $q$ we are still able to determine the exact values of $ASP(k,4,q)$ and $ASB(k,4,q)$ for certain values of $k$:

\begin{proposition} \label{prop:t=4_q=2} 
    For $k \in S =\{1,2,3,4,5,7,8,11,12,16\}$ and $\ell \in \Z_{>0}$ we have 
    $$
        ASP(k,4,2^\ell)=ASB(k,4,2^\ell)= k+1 + \left\lceil \frac{1 + \sqrt{1 + 8k}}{2}\right\rceil,
    $$
    i.e., the lower bound from Theorem \ref{thm:t=4} is the true value for $k \in S$ when $q$ is even. 
\end{proposition}

\begin{proof}
    First, we prove that the parity column of $G'$ has four disjoint recovery sets if and only if $k \in S$. As we can take the column itself as one of the recovery sets, we need to find three additional recovery sets among the columns of $G$. Since each row of $G$ has weight exactly three, every column must belong to some recovery set and any two columns with overlapping support must be in disjoint recovery sets. 
    
    Consequently, the columns of $A$ must be partitioned into three sets $A_1, A_2, A_3$ such that no two columns within the same set overlap in support. If such a partition exists, then supplementing each $A_i$ with appropriate columns of $I_k$ yields three disjoint recovery sets for the parity column.
    
    Let $|A_i| = r_i$. Then
    $$
        r_1 + r_2 + r_3 = r = \left\lceil \frac{1 + \sqrt{1 + 8k}}{2}\right\rceil.
    $$
    
    For any of the $\binom{r_i}{2}$ pairs of columns in $A_i$ there is a weight-$2$ vector from $\F_2^r$ that cannot be used as a row in the construction of $A$ (as otherwise, the support of these two rows will overlap). This means that we must have 
    $$
        \binom{r}{2} - \left( \binom{r_1}{2} + \binom{r_2}{2} + \binom{r_3}{2}\right) \geq k,
    $$

    which can be rewritten as 
    $$
       r^2 - (r_1^2 + r_2^2 + r_3^2) \geq 2k.
    $$

    By the QM-AM inequality this implies 
    $$
        \frac{1}{3} \left\lceil \frac{1 + \sqrt{1 + 8k}}{2}\right\rceil^2 = \frac{1}{3}r^2 \geq k,
    $$
    and one can check that this holds only for $k \in S$. 

    Finally, for each $k \in S$ it can easily be checked that one can choose suitable subsets of columns of $A$.

    Next, we show that there are also four disjoint recovery sets for each column from $A$. In fact, we will show that there is a parition of the columns of $G'$ into four disjoint sets such that each of these is a recovery set, like for the columns of $I_k$.

    Suppose $\bfg$ is a column of $A$.
    By the construction of $A$, the sum of all the other columns of $A$ is equal to $g$ (here we use the fact that $q$ has characteristic $2$).  Hence, one may take as recovery sets: the column $\bfg$ itself; all other columns of $A$; the columns of $I_k$ whose supports intersect the support of $\bfg$; and, finally, the remaining columns of  $=I_k$ together with the parity column.   

    To go from the all-symbol PIR property to the all-symbol batch property, the only non-trivial thing we need to check is that $G'$ can serve a list with two copies of two distinct columns, not both in $I_k$. If the parity column does not appear, then we can use the same argument as in the proof of Lemma \ref{Lem:G'IsBatch}. So, suppose instead that there are two copies of the parity check column together with two copies of another column, say $\bfg_{i_1}$. 
    
    We know that $\bfg_{i_1}$ has four recovery sets that form a partition of the columns of $G'$. Let $R_1,R_2,R_3,R_4\subseteq [n+1]$ be the sets of indices that correspond to these recovery sets. We may assume, without loss of generality, that $R_1 = \{i_1\}$ and that $n+1\in R_4$.

    Now, 
    $$
        \bfg_{n+1} = \sum_{i \in [n]} \bfg_i  = \bfg_{i_0} + \sum_{i\in R_2}\bfg_i + \sum_{j\in R_3}\bfg_j + \sum_{l \in R_4\setminus\{n+1\}} \bfg_l = 3\cdot\sum_{j\in R_3}\bfg_j + \sum_{l \in R_4\setminus\{n+1\}} \bfg_l,
    $$
    so we can take the columns corresponding to $R_1$ and $R_2$ as recovery sets for $\bfg_{i_1}$, and $\{n+1\}$ and $R_3 \cup R_4\setminus \{n+1\}$ as recovery sets for $\bfg_{n+1}$. This finishes the proof.
\end{proof}

\begin{remark}
The above proof shows that $G'$ can serve any list of four columns from $G$ when $q$ is a power of $2$. This is not the case in general, i.e., there is no hope that the proof can be modified to apply for all $q$. However, the part of the proof that shows $G'$ does not satisfy the $4$-all-symbol batch property for $k \not \in S$ holds for any $q$, not just in characteristic 2.
\end{remark}

In general, we are not able to determine the exact values of $ASP(k,4,q)$ and $ASB(k,4,q)$. However, we computed several specific cases by exhaustive computer search. For instance, we found an explicit $5 \times 10$-matrix over a field of characteristic three that satisfies the $4$-all-symbol batch property; see Example~\ref{ex:asp(5,4,3)}. This implies that $ASP(5,4,3^\ell) = ASB(5,4,3^\ell) = 10$, i.e., the lower bound is met in this case. We also verified computationally that $ASP(6,4,2)  > 11$, and hence 
$$
    ASP(6,4,2) = ASB(6,4,2) = 12.
$$
In particular, this shows that the known relation $B(k,4,2) = P(k,4,2) = P(k,3,2) + 1$ does not extend directly to the all-symbol PIR/batch setting.

\begin{example}\label{Ex:t=4}
We present the construction for dimensions $5$ and $6$. For $k=5$, the matrix~$G'$ takes the following form:
\begin{align}
G' =
\left(
\begin{array}{ccccc|cccc|c}
1 & 0 & 0 & 0 & 0 &  -1 & -1 & 0 & 0 & 1 \\
    0 & 1 & 0 & 0 & 0  & -1 & 0 & -1 & 0 & 1 \\
     0 & 0 & 1 & 0 & 0  & -1 & 0 & 0 & -1 & 1 \\
      0 & 0 & 0 & 1 & 0  & 0 & -1 & -1 & 0 & 1 \\
       0 & 0 & 0 & 0 & 1  & 0 & -1 & 0 & -1 & 1  
\end{array}
\right).
\end{align}
\vspace{-0.7cm}
\[
\hspace{1cm}
\underbrace{\hspace{2.6cm}}_{I_5} 
\hspace{0.2cm}
\underbrace{\hspace{3cm}}_{A}
\hspace{0.2cm}
\underbrace{\hspace{-90cm}}_{P}
\]

One can verify that when $q$ has characteristic $2$, 
$G'$ satisfies the $4$-all-symbol batch property, as stated 
in Proposition~\ref{prop:t=4_q=2}. Conversely, if the 
characteristic of $q$ is not $2$, it can be verified that 
the sixth column (i.e., the first column of $A$) does not have 
four disjoint recovery sets.

For the case $k=6$, a single parity column is insufficient, even when $q$ is a power of $2$. In this case, the matrix $G''$ takes the following form:
\begin{align}
G'' =
\left(
\begin{array}{cccccc|cccc|cc}
1 & 0 & 0 & 0 & 0 &  0&-1 & -1 & 0 & 0 & 1 & 1\\
    0 & 1 & 0 & 0 & 0 &0 & -1 & 0 & -1 & 0 & 1 & 1\\
     0 & 0 & 1 & 0 & 0 & 0 & -1 & 0 & 0 & -1 & 1 & 1\\
      0 & 0 & 0 & 1 & 0  & 0 & 0 & -1 & -1 & 0 & 1 & 1\\
       0 & 0 & 0 & 0 & 1  & 0 & 0 & -1 & 0 & -1 & 1 & 1\\
       0 & 0 & 0 & 0 & 0  & 1 & 0 & 0 & -1 & -1 & 1 & 1
\end{array}
\right).
\end{align}
\vspace{-0.7cm}
\[
\hspace{0.8cm}
\underbrace{\hspace{3cm}}_{I_6} 
\hspace{0.3cm}
\underbrace{\hspace{3cm}}_{A}
\hspace{0.3cm}
\underbrace{\hspace{0.8cm}}_{P}
\]
\end{example}

\begin{example}\label{ex:asp(5,4,3)}
Let $q$ be a power of $3$. Then, the following $5 \times 10$-matrix, which we found using a computer search, satisfies the $4$-all-symbol batch property:
$$
\left(
\begin{array}{ccccc|ccccc}
    1 & 0 & 0 & 0 & 0  & 0 & 2 & 1 & 0 & 1 \\
    0 & 1 & 0 & 0 & 0  & 1 & 0 & 0 & 1 & 2 \\
    0 & 0 & 1 & 0 & 0  & 1 & 0 & 2 & 0 & 2 \\
    0 & 0 & 0 & 1 & 0  & 2 & 1 & 0 & 2 & 0 \\
    0 & 0 & 0 & 0 & 1  & 0 & 2 & 1 & 1 & 0 
\end{array}
\right).
$$
\end{example}

\section{ASP and ASB Properties of Codes with Fixed Parameters}\label{Sec:Propetries_Of_Known_Codes}

In this section, we study how the parameters of a linear code influence its all-symbol PIR and all-symbol batch properties. More precisely, we give bounds on the number $t$ for which a code can be $t$-all-symbol PIR or $t$-all-symbol batch, depending on its length, dimension, minimum distance, and dual minimum distance. Moreover, we investigate how two families of codes (MDS and simplex codes) perform with respect to these bounds. In passing, we establish new cases for an open conjecture from~\cite{YAAKOBI2020}.

\subsection{General Bounds}

We will repeatedly use the following well-known result concerning recovery sets. Recall that a codeword $\bfx \in \mC$ is \textit{minimal} if its support is minimal with respect to inclusion.
\begin{lemma} \label{lem:dual}
Let $\mC \le \F_q^n$ be a code with generator matrix $G \in \F_q^{k \times n}$ and let $\mC^\perp \le \F_q^n$ be its dual. The minimal recovery sets (with respect to inclusion) of size larger than one for the $i$-th column of $G$ are in one-to-one correspondence with minimal codewords $\bfx \in \mC^\perp$ with $i \in \supp(\bfx)$.
\end{lemma}

We begin with the following result on $t$-all-symbol PIR codes. This bound is classical in the context of one-step majority-logic decoding and appears, for example, in~\cite[Theorem 8.1]{lin2001error}. Although the result is well known, we include a proof for completeness and to help the reader’s understanding. 

\begin{proposition} \label{prop:tasp}
    Let $\mC \le \F_q^n$ be a code with $d^\perp:=d(\mC^\perp) >1$. If $\mC$ is a $t$-all-symbol PIR code, then
    \begin{align*}
        t \le \left\lfloor\frac{n-1}{d^\perp-1}\right\rfloor+1.
    \end{align*}
\end{proposition}
\begin{proof}
Without loss of generality suppose we want to recover $\bfg_1$ $t$ times. We can use $\bfg_1$ once, and then we need to find $t-1$ disjoint recovery sets for $\bfg_1$. By Lemma~\ref{lem:dual} this is equivalent to asking for $t-1$ codewords $\bfx_1,\dots, \bfx_{t-1} \in \mC^\perp$ with $\supp(\bfx_i) \cap \supp(\bfx_j) = \{1\}$ for all $i \ne j$ and $i,j \in [t-1]$. The minimal cardinality of the support of an element in $\mC^\perp$ is $d^\perp$. Thus, we need
\begin{align*}
    (t-1)(d^\perp-1) \le n-1.
\end{align*}
The inequality comes from the fact that without the coordinate 1, the codewords all have support of size at least $d^\perp-1$, they need to be disjoint outside of 1, and they can \emph{cover} at most $n-1$ coordinates. The statement of the proposition follows.
\end{proof}

We recall the definitions of shortening and puncturing of a code.

\begin{definition} \label{short/punct}
    Let $\mC \leq \F_q^n$ be a linear code and $A \subseteq [n]$. 
    \begin{itemize}
        \item[(i)] $\supp(\bfx):=\{i : x_i \ne 0\}$ denotes the \textbf{support} of $\bfx=(x_1,\dots,x_n) \in \F_q^n$;
        \item[(ii)] $\mC(A):=\{\bfx \in \mC : \supp(\bfx) \subseteq A\}$ is the \textbf{shortening} of $\mC$ by the set $A$;
        \item[(iii)] $\pi_{A}(\mC):=\{\pi_{A}(\bfx) : \bfx \in \mC\}$, where $\pi_A:\F_q^n \to \F_q^{|A|}$ is the projection onto the coordinates indexed by $A$, is the \textbf{puncturing} of $\mC$ by the set $A$.
    \end{itemize}
\end{definition}

The bound from Proposition~\ref{prop:tasp} clearly also holds for $t$-all-symbol batch codes. However, tweaking the statement of Proposition~\ref{prop:tasp} to match the property of being $t$-all-symbol batch, we obtain the following stronger bound.

\begin{proposition}\label{cor:tasb}
Let $\mC \le \F_q^n$ be a code with $d^\perp:=d(\mC^\perp) >1$. If $\mC$ is an $t$-all-symbol batch code, then for all $1 \le s \le n-1$ we have
\begin{align*}
        t \le \left\lfloor\frac{n-s}{\max\{d(\mC^\perp(S)) : S \subseteq [n], |S|=n-s+1\}-1}\right\rfloor+s.
    \end{align*}
\end{proposition}

\begin{proof}
Suppose we want to recover a total of $1 \le s \le t \le n-1$ different columns which are indexed by $S$, where $s < t$ means we want to recover at least one of them more than once. Suppose we want to recover one of the $s$ columns $t-s+1$ times and all other columns once. Then, apart from the $s$ columns we use as recovery sets of size 1, we need $t-s$ disjoint recovery sets for that specific column. Without loss of generality say the column we want to recover is the first one. For this, we need $t-s$ codewords $\bfx_1,\dots,\bfx_{t-s} \in \mC^\perp$ with $\supp(\bfx_i) \cap \supp(\bfx_j) = \{1\}$ for all $i \ne j$ and $i,j \in [t-s]$. In particular, since the recovery sets have size at least $d^\perp-1$, and the codewords $\bfx_1,\dots,\bfx_{t-s}$ need to be contained in $\mC^\perp(S)$, we get
\begin{align*}
    (t-s)\left(d(\mC^\perp(S))-1\right) \le n-s.
\end{align*}
Since this has to hold for all $1 \le s \le t \le n-1$ and all $S \subseteq [n]$ with $|S|=s$ we obtain the upper bound of the proposition.
\end{proof}

By setting $s=1$ in Corollary~\ref{cor:tasb} we recover Proposition~\ref{prop:tasp}.

\begin{remark}
For a code $\mC \le \F_q^n$ with generator matrix $G$ with repeated columns, the bound of Proposition~\ref{prop:tasp} can only be attained in the case where $\mC$ is the repetition code. This is because if there are repeated columns, then $d^\perp=2$ and so the upper bound of Proposition~\ref{prop:tasp} can only be attained if $t=n$. Therefore, any column of $G$ can be recovered $n$-times with recovery sets of size $1$, implying that $G=(\bfg,\dots,\bfg) \in \F_q^{1 \times n}$ for some $\bfg \in \F_q$, and $\mC$ is the repetition code.
\end{remark}

\begin{proposition} \label{prop:jon}
Let $\mC \le \F_q^n$ be a $t$-all-symbol PIR code where any list $\{\bfg_i^t\}$ for $i \in [n]$ can be served with recovery sets of size at most $r$, for some $r \le n/t$. Then $\mC$ is $(\lfloor t/r\rfloor+1)$-all-symbol batch.
\end{proposition}
\begin{proof}
Let $\bfg_{a_1},\dots,\bfg_{a_{\lfloor t/r\rfloor+1}}$ be the list of columns that we want to serve. Suppose we have already chosen disjoint recovery sets for $\bfg_{a_1},\dots,\bfg_{a_\ell}$ of size at most $r$ for some $0 \le \ell < \lfloor t/r\rfloor+1$. At least one of the recovery sets can be chosen as the column itself, and so the number of columns that are being used so far is at most $r(\ell-1)+1$. By the assumption that $\mC$ is $t$-all-symbol PIR, we have $t$ recovery sets of size at most $r$ for $\bfg_{a_{\ell+1}}$. Each column used for recovering $\bfg_{a_1},\dots,\bfg_{a_\ell}$ can be in at most one of the $t$ recovery sets of $\bfg_{a_{\ell+1}}$. We have used strictly less than $t$ columns since
\begin{align*}
    r(\ell-1)+1 \le r\left(\lfloor t/r\rfloor-1\right)+1 < t,
\end{align*}
so it follows from the Pigeonhole principle that there is at least one of the $t$ recovery sets of $\bfg_{a_{\ell+1}}$ which contains only unused columns. This means that we can recover $\bfg_{a_{\ell+1}}$. We can repeat this process until we found all $\lfloor t/r\rfloor$ recovery sets.
\end{proof}

\subsection{MDS Codes}
Suppose $\mC \le \F_q^n$ is an MDS code of dimension $k$. Since the dual of an MDS code is also MDS, we have $d^\perp=k+1$ and the bound of Proposition~\ref{prop:tasp} reads as
\begin{align*}
    t \le \left\lfloor\frac{n-1}{d^\perp-1}\right\rfloor+1=\left\lfloor\frac{n-1}{k}\right\rfloor+1.
\end{align*}
To explicitly compute $t$ such that $\mC$ is $t$-all-symbol PIR, note that in a generator matrix $G$ of $\mC$, any~$k$ columns are linearly independent. In particular, a column can be recovered only from the recovery set of size one, or any set of $k$ different columns. This gives that $\mC$ is $t$-all-symbol PIR with
\begin{align*}
    t=\left\lfloor\frac{n-1}{k} \right\rfloor+1.
\end{align*}
If we look at $\mC$'s all-symbol batch properties, we see that serving a number of requests made of different columns is \textit{easier} than serving a single column the same number of times. More precisely, suppose we have set $\{\bfg_1^{t_1},\dots,\bfg_{\ell}^{t_{\ell}}\}$ of requests made of the columns of $G$. We can allocate to all requests a recovery set of size one, given by the corresponding column, and then need any set of $k$ different columns for all remaining requests. Therefore $\mC$ is also $t$-all-symbol batch. 

We conclude that the bound of Proposition~\ref{prop:tasp} is met with equality for all $k$ and $n$.

\subsection{Simplex Code}
Suppose $\mC \le \F_2^n$ is the simplex code of dimension $k$ and length $n=2^k-1$. The dual of the simplex code is the Hamming code, which has minimum distance $d^{\perp}=3$. Therefore the bound of Proposition~\ref{prop:tasp} reads
\begin{align*}
    t \le \left\lfloor\frac{n-1}{d^\perp-1}\right\rfloor+1=\left\lfloor\frac{2^k-2}{3-1}\right\rfloor+1=2^{k-1}.
\end{align*}
Since the columns of the generator matrix $G$ of $\mC$ are all non-zero vectors in $\F_2^k$, apart from the recovery set of size 1, for every fixed column one can partition the vectors of $\F_2^k$ into two-sets, where each two-set is a recovery set for that fixed column. Because of this, the above bound is met with equality. 

In addition, the following result was proved in~\cite[Lemma 12]{wang2017switch}.

\begin{theorem} \label{thm:simba}
The simplex code $\mC \le \F_2^n$ of dimension $k$ and length $n=2^k-1$ is a $2^{k-1}$-batch code.
\end{theorem}

If we consider the all-symbol batch property, then the question of whether $\mC$ is a $t$-all-symbol batch code for $t=2^{k-1}$ remains open. This was posed as a conjecture in~\cite{YAAKOBI2020}:

\begin{conjecture} \label{conj:eit}
    The simplex code $\mC \le \F_2^n$ of dimension $k$ and length $n=2^k-1$ is a $2^{k-1}$-functional batch code.
\end{conjecture}

Note that since all vectors in $\F_2^k$ appear as columns of the generator matrix of $\mC$, the property of being a $t$-functional batch code is equivalent to being a $t$-all-symbol batch code.

Some progress toward resolving this conjecture has been made in recent years. The best currently known results, to the best of our knowledge, show that the simplex code is a $t$-functional batch code for $t=\left\lfloor \frac{2}{3}\cdot 2^{k-1}\right\rfloor $
by~\cite{yohananov2021}, and for
$
t=\left\lfloor \frac{5}{6}\cdot 2^{k-1}-k \right\rfloor$
by~\cite{yohananov2022}. In addition, it was shown in~\cite{hollmann2023} that the simplex code is a $t$-odd batch code for $t=2^{k-1}$, where an odd batch code refers to the case of serving only vectors of odd weight.

By Lemma~\ref{Invariance_Of_Mul_By_Invertible} and Lemma~\ref{lem:dual}, we derive the following result, which covers additional cases and thus provides further evidence in support of Conjecture~\ref{conj:eit}.

\begin{proposition}
    The simplex code $\mC \le \F_2^n$ of dimension $k$ and length $n=2^k-1$ can serve any list $\{\bfg_1^{t_1},\dots,\bfg_\ell^{t_{\ell}}\}$ with $t_1+\dots+t_{\ell} = 2^{k-1}$ where $\bfg_1,\dots,\bfg_\ell \in \F_2^k$ are such that $\dim(\langle\bfg_1,\dots,\bfg_\ell\rangle)=\ell$. 
\end{proposition}
\begin{proof}
    Let $G = (I_k \mid A) \in \mathbb{F}_2^{k \times n}$ be a systematic generator matrix of the code $\mathcal{C}$. By Theorem~\ref{thm:simba}, any multiset of requests involving only the first $k$ columns of $G$ can be served. By Lemma~\ref{lem:dual}, the dual code $\mathcal{C}^\perp$ is invariant under changes of the generator matrix. Therefore, this ability to serve requests made from the first $k$ columns holds for any generator matrix of $\mathcal{C}$. In particular, for any invertible matrix $M \in \mathbb{F}_2^{k \times k}$, the matrix $MG = (M, MA)$ is also a generator matrix of $\mathcal{C}$. In this representation, the first $k$ columns correspond to the rows of $M$, which are linearly independent. Thus, any multiset of requests involving linearly independent vectors can be served by $\mathcal{C}$.
\end{proof}

\section{Discussion and Future Directions}\label{Sec:Conclustion}

{In this paper, we study codes with the property that $t$ (not necessarily distinct) symbols of a codeword can be recovered from pairwise disjoint sets of codeword symbols. We distinguish two settings: recovering the same symbol $t$ times, leading to \emph{$t$-all-symbol PIR codes}, and recovering an arbitrary multiset of $t$ symbols, leading to \emph{$t$-all-symbol batch codes}. These notions unify and generalize several previously studied code properties, including one-step majority-logic decodable codes, (functional) PIR codes, and (functional) batch codes.
Our main contributions are the following: we determine the minimum length required for a code of fixed dimension to satisfy these properties for some small values of $t$, we characterize structural properties of the generator matrices of codes achieving this optimal length, and we provide bounds and insights into how well a code with fixed length, dimension, and other parameters can satisfy these recovery requirements. While we make progress towards the understanding of these code families, a number of interesting questions remain open:
\begin{enumerate}
\item In this work we determine the minimum length of $t$-all-symbol PIR and batch codes for small values of $t$ (namely $t \in \{1,2,3\}$ and partial results for $t=4$). It remains open to characterize, or at least bound, the minimum length of optimal codes with small dimension $k$. It would be particularly interesting to understand whether the behavior for small $k$ aligns with that of standard PIR and batch codes, or whether additional redundancy is needed to achieve the all-symbol recovery property.
\item Another natural direction is to study the asymptotic behavior of the minimum length of optimal all-symbol PIR and batch codes as either $t$ or $k$ grows.
\item Our results for $t=4$ still leave a gap; future work could focus on determining exact values in this case.
\item Most of our results for $ASB(k,t,q)$ and $ASP(k,t,q)$ do not depend on the alphabet size $q$. Intuitively, however, one expects that increasing $q$ should lead to shorter $t$-all-symbol batch or PIR codes for a fixed dimension $k$. Understanding how $q$ influences $ASB(k,t,q)$ and $ASP(k,t,q)$ (for example through bounds that take into account $q$) is an interesting direction for future work.
\item Finally, investigating the $t$-all-symbol PIR and batch properties for additional families of well-known codes (such as Hamming and Reed--Muller codes) remains open. In particular, although our approach resolves further cases of Conjecture~\ref{conj:eit}, several instances of the conjecture remain unresolved.
\end{enumerate}
}

\bibliographystyle{ieeetr}
\bibliography{ourbib}
\clearpage

\end{document}